\documentclass[conference,letterpaper,romanappendices]{ieeeconf}
\usepackage{graphics}
\usepackage{cite}
\usepackage[pdftex]{graphicx}
\usepackage{epstopdf}
\usepackage{epsfig}
\usepackage{latexsym}
\usepackage{amsfonts}
\usepackage{calc}
\usepackage{url}
\usepackage{enumerate}
\usepackage{color}
\usepackage[tbtags]{amsmath}
\usepackage{amssymb}
\usepackage{upref}
\usepackage{dsfont}
\usepackage{multirow}
\usepackage{booktabs}
\usepackage{bigstrut}
\usepackage{rotating}
\usepackage{mathtools,bbm}
\usepackage{arydshln}
\usepackage{breqn}

\newtheorem{theorem}{Theorem}

\newtheorem{lemma}{Lemma}

\newtheorem{definition}{Definition}

\renewcommand{\QED}{\QEDopen}

\IEEEoverridecommandlockouts

\begin{document}
\title{\Huge Cooperative Data Exchange with Unreliable Clients}

\author{Anoosheh Heidarzadeh and Alex Sprintson\thanks{The authors are with the Department of Electrical and Computer Engineering, Texas A\&M University, College Station, TX 77843 USA (E-mail: anoosheh@tamu.edu; spalex@tamu.edu).}\thanks{This work was supported by the National Science Foundation under Grant No. CNS-0954153 and by the AFOSR under contract No. FA9550-13-1-0008.}}

\maketitle
\thispagestyle{empty}
\begin{abstract} 
Consider a set of clients in a broadcast network, each of which holds a subset of packets in the ground set $X$. In the (coded) cooperative data exchange problem, the clients need to recover all packets in $X$ by exchanging coded packets over a lossless broadcast channel. Several previous works analyzed this problem under the assumption that each client initially holds a random subset of packets in $X$. In this paper we consider a generalization of this problem for settings in which an unknown (but of a certain size) subset of clients are unreliable and their packet transmissions are subject to arbitrary erasures. For the special case of one unreliable client, we derive a closed-form expression for the minimum number of transmissions required for each reliable client to obtain all packets held by other reliable clients (with probability approaching $1$ as the number of packets tends to infinity). Furthermore, for the cases with more than one unreliable client, we provide an approximation solution in which the number of transmissions per packet is within an arbitrarily small additive factor from the value of the optimal solution.
\end{abstract}

\section{Introduction}
Consider a network of clients that share a broadcast channel, each of which holds a subset of packets of the ground set $X$ of size $K$. In the \emph{cooperative data exchange problem}~\cite{SSBR:2010} (also known as \emph{universal recovery}), each client wishes to recover all the packets in $X$. To achieve this goal, the clients exchange data by transmitting coded packets over a shared lossless broadcast channel. Assuming that each client knows which packets are known by all other clients, the problem is to specify how many and which (coded) packets each client requires to transmit so as to achieve the universal recovery. 

In this work, we consider a generalization of the cooperative data exchange problem, for the settings where a certain number of clients are unreliable and their packet transmissions are subject to arbitrary erasures. Specifically, our problem is to minimize the total number of transmissions required to achieve \emph{robust recovery}, i.e., each reliable client can recover all packets held by the other reliable clients. Since the identity of the unreliable clients is unknown, the coding scheme must include redundant transmissions to tolerate a failure of a subset of clients of a certain size.

This problem has several interesting practical applications. For instance, it captures the scenario where some clients, initially part of the network, leave the network (deliberately or not) before the end of the data exchange session. Another example is the scenario where a subset of clients are compromised by an adversary, and accordingly their packet transmissions can be dropped in an arbitrary manner.

\subsection{Related Work}
Recently, there has been a significant interest in the cooperative data exchange problem, specifically due to the emergence of powerful techniques employing network coding~\cite{ACLY:2000,LYC:2003}. The cooperative data exchange problem was originally introduced in~\cite{RCS:2007}, where a broadcasting network was considered, and was later generalized to arbitrary networks in~\cite{CXW:2010,CW:2010,CW:2014,GL:2012}. Originally, lower and upper bounds on the minimum required number of transmissions were established in~\cite{RSS:2010}, and later, randomized and deterministic solutions to the problem were presented in~\cite{SSBR:2010,SSBR2:2010} and \cite{MPRGR:2011}. Scenarios considering various transmission costs have been studied in~\cite{OS:2011,TSS:2011}, and scenarios providing secrecy and weak security, in the presence of an eavesdropper, have been considered in~\cite{CW:2011,CH:2014} and~\cite{YS:2013,YSZ:2014}, respectively.  

To the best of our knowledge, the only ``closed-form'' solution to the cooperative data exchange problem (without unreliable clients) is given in~\cite{CXW:2010}, under the assumption of random packet distribution. This solution is shown to be correct with probability approaching $1$ as the number of packets approaches infinity. Such a result, while asymptotic and exclusive to the random packet distribution, provides valuable theoretical insights as well as reasonable approximation that can be used for constructing practical algorithms. However, this solution is limited to the settings in which all clients are reliable. This motivates the present work which attempts to bridge the gap and investigates closed-form (exact and approximate) solutions to the cooperative data exchange problem with unreliable clients, under the random packet distribution assumption.  


\subsection{Our Contributions}
For the case with an arbitrary number of unreliable clients, we compute a closed-form \emph{approximate solution} in which the total number of transmissions per packet (i.e., normalized by the number of packets $K$) is within an arbitrarily small (yet non-vanishing) additive factor of the optimal solution, with probability approaching $1$ as $K$ goes to infinity. 

Also, for the special case with one single unreliable client, we derive a closed-form \emph{exact solution} which requires, with probability approaching $1$ as $K$ goes to infinity, the minimum total number of transmissions. The exact solution yields a zero additive optimality gap (to the minimum total number of transmissions) and hence is stronger than our approximate solution which yields a nonzero gap that grows linearly with $K$. The strength of this result, however, comes with its restriction to a special case, and its generalization to settings with more than one unreliable client remains open.


\section{Problem Setup and Definitions}\label{sec:ProbForm}
Consider $N$ clients and the set $X$ of $K$ packets $x_1,x_2,\dots,x_K$. We use the short notation $[n]$ to represent the set $\{1,\dots,n\}$, for any integer $n$. Each client $i\in [N]$ holds a subset $X_i$ of the packets in the set $X$ (without loss of generality, we assume $X=\cup_{1\leq i\leq N} X_i$). We also denote by $\overline{X}_i = X\setminus X_i$ the set of packets missing at the client $i$. We further assume that each packet is available at each client, independently from other packets and clients, with probability (w.p.) $\alpha$. (This assumption is referred to as the \emph{random packet distribution} in~\cite{CXW:2010}.) 

We assume that $M$ ($0\leq M<N$) clients are ``unreliable'' and that the identity of unreliable clients is unknown. Each reliable client broadcasts over a lossless channel, while the packets broadcasted by each unreliable client are subject to arbitrary erasures. The goal of the cooperative data exchange in this setting is to achieve \emph{robust recovery} that guarantees that each reliable client can recover all the packets known by the other reliable clients. We use the notion of robust recovery since universal recovery might not be achievable in our setting due to the fact that that it might not be possible to obtain a packet held by unreliable clients only. It is worth noting that for the special case with no unreliable clients ($M=0$), the robust recovery problem becomes equivalent to the universal recovery problem. 

We further assume that each packet $x_i$ is $P$-divisible, where $P=N-M-1$, i.e., $x_i$ can be partitioned into $P$ chunks of equal size (the reason for this choice of $P$ will become clear later), and transmissions can consist of a single chunk (as opposed to an entire packet).

Let $\{r_i\}\doteq \{r_i\}(\{X_j\})$, $1\leq i\leq N$ be the \emph{transmission schedule} for a given instance $\{X_i\}$ of the problem at hand. (By the $P$-divisibility assumption, it follows that for each client $i$, the number of its transmissions $r_i$ is a rational number of the form $\frac{n}{P}$, for some non-negative integer $n$). The transmission schedule $\{r_i\}$ is said to be \emph{feasible for instance} $\{X_i\}$ if there exists a coding scheme with each client $i$ transmitting $r_i$ coded packets that achieves robust recovery. The following definitions assume that $\{X_i\}$ is drawn according to the random packet distribution.

\begin{definition}
The transmission schedule $\{r_i\}$ is said to be \emph{feasible} if it is feasible for a random instance $\{X_i\}$ w.p.~approaching $1$ as $K\rightarrow\infty$. 
\end{definition}

\begin{definition}
The transmission schedule $\{r_i\}$ is said to be an \emph{exact solution} if it is feasible and $\sum_{i} r_i$ is equal to the minimum total number of transmissions required for robust recovery.
\end{definition}

\begin{definition}
The transmission schedule $\{r_i\}$ is said to be an \emph{approximate solution} if it is feasible and $\frac{1}{K}\sum_{i} r_i$ is within $\epsilon$ of the ratio of the minimum total number of transmissions required for robust recovery to $K$, for any $\epsilon>0$, w.p.~approaching $1$ as $K\rightarrow\infty$. 
\end{definition}

In this work, our problem is to determine a closed-form exact or approximate solution for a random instance of the robust recovery problem. Given a feasible transmission schedule, the clients can achieve robust recovery (with high probability) by employing random linear network coding (over a sufficiently large finite field), i.e., transmitting random linear combinations of their packets. (This comes from the fact that the problem of robust recovery, similar to the problem of universal recovery~\cite{CXW:2010,CW:2010,CW:2014}, can be reduced to a multicast network coding problem.) 

\section{Main Results}\label{sec:MainResults}

For the special case with no unreliable client ($M=0$), it was previously shown in~\cite{CXW:2010} that the optimal number of transmissions for each client can be found by solving the following Linear Program (LP): 
\begin{eqnarray}\label{eq:LP}
\mathrm{minimize} && \sum_{i=1}^{N} r_i, \\ \nonumber
\mbox{s.t.} && \sum_{i\in\cal{N}} r_i\geq \bigg|\bigcap_{i\in {\mathcal{\overline{N}}}}{\overline{X}}_i\bigg|, \hspace{5pt} \forall \emptyset \subsetneq{\mathcal{N}}\subsetneq [N],
\end{eqnarray} where $\overline{\mathcal{N}} = [N]\setminus \mathcal{N}$. 

Now, consider the case with $M$ unreliable clients. Since the set of unreliable clients ($\mathcal{I}$) is not known apriori, the robust recovery is achievable so long as for every $\mathcal{I}\subset [N]$, $|\mathcal{I}|=M$, each client $i\in[N]\setminus\mathcal{I}$ can recover all the packets held by the other clients $j\in[N]\setminus\mathcal{I}$. In the case without unreliable clients, the set of packets each client $i$ requires, $\overline{X}_i$, is the collection of the packets available at the other clients (but not available at client $i$), i.e.,
\begin{equation}
\overline{X}_i = X\setminus X_i.	
\end{equation} However, in the presence of unreliable clients, the set of packets each reliable client $i\notin \mathcal{I}$ requires, denoted by $\overline{X}_{i,\mathcal{I}}$, is the collection of the packets each of which is available at some other reliable client (but not available at client $i$), i.e., 
\begin{equation}\label{eq:XiI}
\overline{X}_{i,\mathcal{I}}= \cup_{j\in [N]\setminus \mathcal{I}} X_j \setminus X_i. 
\end{equation} Thus, we need to revise the set of constraints in~\eqref{eq:LP} so as to take into account (i) every possible set of unreliable clients and (ii) the set of packets each reliable client requires for any possible subset of unreliable clients. The following theorem is a straightforward generalization of LP~\eqref{eq:LP} for the case with $M$ unreliable clients, and appears without proof. 

\begin{theorem}[Robust Recovery]\label{thm:NewLP} 
The minimum total number of transmissions required for robust recovery is the optimal value of the following LP:
\begin{eqnarray}\label{eq:ExpNewLP}
\mathrm{minimize} && \sum_{i=1}^{N} r_i, \\ \nonumber
\mbox{s.t.} && \sum_{i\in\cal{N}} r_i\geq \max_{{\scriptsize \begin{array}{c} \mathcal{I}\subset \mathcal{\overline{N}},\\ |\mathcal{I}|=M\end{array}}}\Biggl|\bigcap_{i\in {\cal \overline{N}}\setminus\mathcal{I}}{\overline{X}}_{i,\mathcal{I}}\Biggr|, \\ \nonumber && \forall \{\mathcal{N}\subset [N]:1\leq |\mathcal{N}|\leq P\},
\end{eqnarray}
 where $N$ is the number of clients, $M$ is the number of unreliable clients.
\end{theorem}


Our goal is to solve LP~\eqref{eq:ExpNewLP}. It is noteworthy that LP~\eqref{eq:LP}, which is a special case of LP~\eqref{eq:ExpNewLP} when $M=0$, was previously given a closed-form exact solution in~\cite{CXW:2010}: 

\begin{theorem} \cite[Theorem~4]{CXW:2010} \label{thm:LPSolution} 
$\{\tilde{r}_i\}$ is an exact solution to LP~\eqref{eq:LP}: 
\[
\tilde{r}_i = \sum_{j=1}^{N}\frac{1}{N-1} \left|\overline{X}_j\right|-\left|\overline{X}_i\right|, \hspace{5pt} 1\leq i\leq N.
\] 
\end{theorem}

\vspace{0.25cm}
The following summarizes (into three steps) the technique which was previously used to solve LP~\eqref{eq:LP}: 
\begin{itemize}
\item[(i)] Choose the set of $N$ constraints in~\eqref{eq:LP} corresponding to the $N$ subsets $\{\mathcal{N}\subset [N]$: $|\mathcal{N}| = N-1\}$.
\item[(ii)] Solve the system of $N$ linear equations corresponding to the $N$ constraints of step (i) for the $N$ unknowns $\{r_i\}$ (where the inequalities are replaced with equality). 
\item[(iii)] Show the feasibility and optimality of the solution of step (ii) with respect to the rest of the constraints in~\eqref{eq:LP}.	
\end{itemize}

Now, a natural question is whether we can use such a deceptively simple, yet remarkably powerful, technique to solve LP~\eqref{eq:ExpNewLP}. The answer is positive, yet as we will show later the two steps (i) and (iii) require a significant amount of non-trivial modifications to become applicable to our problem. The complication comes from the fact that in our case, as opposed to the case with no unreliable clients, a ``proper'' choice of constraints in the step (i), yielding a solution in the step (ii) which is satisfactory with respect to the requirements in the step (iii), is not obvious. Also, it is not clear whether in our case such a proper choice of the constraints always exists. The following theorems summarize our main results. 

For the ease of exposition, we define 
\begin{eqnarray}\label{eq:kidef}
&& k_j = \max\left|\overline{X}_{i,{\mathcal{\overline{N}}\setminus \{i\}}}\right|, \hspace{5pt} 1\leq j\leq N-M,
\end{eqnarray} where the maximization is over all $i\in \mathcal{\overline{N}}$ and all $\{\mathcal{N}\subset [N]:|\mathcal{N}|=P, \{n\}_{1\leq n<j}\in\mathcal{N}\}$, and 
\begin{equation}\label{eq:kijdef}
k_{i,j}=\left|\overline{X}_{i,\{j\}}\right|, \hspace{5pt} 1\leq i\neq j\leq N.	
\end{equation} 


\begin{theorem}[Approximate Solution]\label{thm:OverConstrainedReducedLPSolution}
Assume that the clients are re-labeled such that $k_i\geq k_j$, $1\leq i<j\leq N-M$.\footnote{According to~\eqref{eq:kidef}, it is easy to see that such a re-labeling always exists.} Let $Q\geq 1$ and $0\leq R \leq M$ be some integers such that $N = (M+1)(Q+1)-R$. Then, $\{\tilde{r}_i\}$ is an approximate solution to LP~\eqref{eq:ExpNewLP}:  	
\begin{equation}\label{eq:Tilderis}
\tilde{r}_i=\left\{
\begin{array}{ll}
\tilde{r}-k_i, & 1\leq i\leq Q,\\ 
\tilde{r}-k_{Q+1}, & Q< i\leq N,
\end{array} 
\right.\end{equation} where 
\begin{equation}\label{eq:Tilder}
\tilde{r} = \sum_{i=1}^{Q}\frac{1}{P}k_i+\frac{P-Q+1}{P}k_{Q+1}.
\end{equation} Moreover, 
\begin{dmath}\label{eq:OptimalValuePrimal} \sum_{i=1}^{N} \tilde{r}_i = \sum_{i=1}^{Q} \left(\frac{N-P}{P}\right)k_i+\left(\frac{N+Q(P-N)}{P}\right)k_{Q+1}.
\end{dmath}
\end{theorem}

\begin{proof}The proof is given in Section~\ref{subsec:ApproximateThm}.
\end{proof}

\vspace{0.25cm}
\begin{theorem}[Exact Solution]\label{thm:SCOverConstrainedReducedLPSolution}
Assume that the clients are re-labeled such that $k_{i,j}\geq k_{j,i}$, $1\leq i< j\leq N$.\footnote{The re-labeling procedure is as follows: for each $n$, starting from $1$ and ending at $N-1$, switch the labels of clients $n$ and $n+1$ if and only if $k_{n,n+1}<k_{n+1,n}$. The proof is straightforward and follows from the fact that for any $i,j,l$, if $k_{i,j}\geq k_{j,i}$ and $k_{j,l}\geq k_{l,j}$, then $k_{i,l}\geq k_{l,i}$.} Let $Q\geq 1$ and $0\leq R \leq 1$ be some integers such that $N = 2(Q+1)-R$. Then, $\{\tilde{r}_i\}$ is an exact solution to LP~\eqref{eq:ExpNewLP} when $M=1$: 
\begin{dmath}\label{eq:SCtilder}\tilde{r}_i=\left\{
\begin{array}{ll}
\tilde{r}-k_{i,N}-k_{1,N-Q}, & 1\leq i\leq N-Q,\\ 
\tilde{r}-k_{N-Q,N}-k_{1,i}, & N-Q<i\leq N,
\end{array} 
\right.\end{dmath} where 
\begin{dmath}\tilde{r}=\frac{1}{P}\sum_{i=1}^{N-Q-1} k_{i,N}+\frac{Q}{P}k_{N-Q,N}+\frac{1}{P}\sum_{i=N-Q+1}^{N-1}k_{1,i}+\frac{N-Q-1}{P}k_{1,N-Q}. 
\end{dmath} Moreover, 
\begin{dmath}\label{eq:SCOptimalValuePrimal}\sum_{i=1}^{N} \tilde{r}_i= \frac{2-P}{P}k_{1,N}+\sum_{i=2}^{N-Q-1}\frac{2}{P}k_{i,N}+\frac{R}{P}k_{N-Q,N}+\sum_{i=N-Q+1}^{N-1}\frac{2}{P}k_{1,i}+\frac{2-R}{P}k_{1,N-Q}.
\end{dmath}
\end{theorem}

\vspace{0.25cm}
\begin{proof}The proof is given in Section~\ref{subsec:ExactThm}.
\end{proof}

\section{Proofs}\label{sec:Proofs}
In this section, we give the proofs of theorems~\ref{thm:OverConstrainedReducedLPSolution} and~\ref{thm:SCOverConstrainedReducedLPSolution}. Before giving the proofs, for the ease of exposition we state a few definitions. Consider a generic LP as follows:
\begin{eqnarray}\label{eq:GenLP}
\mathrm{minimize} && \sum_{i=1}^{N} r_i, \\ \nonumber
\mbox{s.t.} && \sum_{i\in\mathcal{N}} r_i\geq f(\{X_i\};\mathcal{N}), \hspace{5pt}\forall \emptyset \subsetneq{\mathcal{N}}\subsetneq [N],
\end{eqnarray} where $f(\{X_i\};\mathcal{N})$ is an arbitrary function of $\{X_i\}$ and $\mathcal{N}$. The following definitions are with respect to LP~\eqref{eq:GenLP}:

\begin{definition}
A sequence $\{r_i\}$ is said to be \emph{feasible} if it satisfies the constraints for a random instance $\{X_i\}$ w.p.~approaching $1$ as $K\rightarrow\infty$. 
\end{definition}

\begin{definition}
A sequence $\{r_i\}$ is said to be \emph{optimal} if $\sum_{i} r_i$ is equal to the optimal value. Moreover, $\{r_i\}$ is said to be a \emph{solution} if it is feasible and optimal. 
\end{definition}

\begin{definition}
A sequence $\{r_i\}$ is said to be \emph{near-optimal} if $\frac{1}{K}\sum_{i} r_i$ is within $\epsilon$ of the optimal value normalized by $K$, for any $\epsilon>0$, w.p.~approaching $1$ as $K\rightarrow\infty$. Moreover, $\{r_i\}$ is said to be an \emph{approximate solution} if it is feasible and near-optimal.  
\end{definition}


\subsection{Proof of Theorem~\ref{thm:OverConstrainedReducedLPSolution}}\label{subsec:ApproximateThm}
Consider a reduced version of LP~\eqref{eq:ExpNewLP} as follows: 
\begin{eqnarray}\label{eq:FinalExpNewLP}
\mathrm{minimize} && \sum_{i=1}^{N} r_i, \\ \nonumber
\mbox{s.t.} && \sum_{i\in\mathcal{N}} r_i\geq \max_{i\in \mathcal{\overline{N}}}\left|\overline{X}_{i,{\mathcal{\overline{N}}\setminus \{i\}}}\right|, \hspace{5pt}\forall\{\mathcal{N}:|\mathcal{N}|=P\}. 
\end{eqnarray} (From now on, we adopt the notation $\mathcal{N}$ to represent an arbitrary subset of $[N]$, unless otherwise stated.)

For arbitrary $M$, no closed-form solution to LP~\eqref{eq:FinalExpNewLP} is known. (However, we will give a closed-form solution to LP~\eqref{eq:FinalExpNewLP} later for the case of $M=1$.) We, instead, give a closed-form solution to LP~\eqref{eq:ReducedExpNewLP}, which we will construct by over-constraining LP~\eqref{eq:FinalExpNewLP}. Next, we show that our solution to LP~\eqref{eq:ReducedExpNewLP} is an approximate solution to LP~\eqref{eq:FinalExpNewLP}, and subsequently, LP~\eqref{eq:ExpNewLP}, which was to be solved ultimately. 

We construct LP~\eqref{eq:ReducedExpNewLP} by replacing $\max_{i\in\mathcal{\overline{N}}}|\overline{X}_{i,{\overline{\mathcal{N}}\setminus\{i\}}}|$ with $k_j$ (given by~\eqref{eq:kidef}) in LP~\eqref{eq:FinalExpNewLP}: 
\begin{eqnarray}\label{eq:ReducedExpNewLP}
\mathrm{minimize} && \sum_{i=1}^{N} r_i, \\ 
\nonumber \mbox{s.t.} && \sum_{i\in\cal{N}} r_i\geq k_j, \hspace{5pt} \forall 1\leq j\leq N-M,\\ 
\nonumber && \forall \{\mathcal{N}\subset [N]\setminus\{j\}:|\mathcal{N}|= P,\{i\}_{1\leq i<j}\in\mathcal{N}\}.
\end{eqnarray} (By~\eqref{eq:kidef}, $k_j\geq \max_{i\in\mathcal{\overline{N}}}|\overline{X}_{i,{\overline{\mathcal{N}}\setminus\{i\}}}|$, for all $\{\mathcal{N}:|\mathcal{N}|=P, \{i\}_{1\leq i<j}\in\mathcal{N}\}$.) 

\begin{lemma}\label{lem:ReducedExpNewLPSolution}
$\{\tilde{r}_i\}$ is an exact solution to LP~\eqref{eq:ReducedExpNewLP}.	
\end{lemma}

\begin{proof} 
We prove the feasibility and the optimality of $\{\tilde{r}_i\}$ to LP~\eqref{eq:ReducedExpNewLP} in lemmas~\ref{lem:Feasibility} and~\ref{lem:Optimality}, respectively.	
\end{proof}

The following lemma is useful in the proof of Lemma~\ref{lem:Feasibility}. 
\begin{lemma}\label{lem:Increasingr}
$\tilde{r}_1\leq \tilde{r}_2\leq \dots\leq \tilde{r}_Q\leq \tilde{r}_{Q+1}=\dots=\tilde{r}_N$.	
\end{lemma}

\begin{proof} By the assumption, 
\begin{equation}\label{eq:kis}
k_1\geq k_2\geq \dots\geq k_{N-M}.
\end{equation} By definition, $\tilde{r}_i=\tilde{r}_{i+1}$, $Q<i<N$. Thus, it remains to show $\tilde{r}_i-\tilde{r}_{i+1}\leq 0$, $1\leq i\leq Q$. By combining~\eqref{eq:Tilderis} and~\eqref{eq:Tilder} along with~\eqref{eq:kis}, $\tilde{r}_i-\tilde{r}_{i+1}=k_{i+1}-k_i\leq 0$, $1\leq i\leq Q$. 
\end{proof}

\begin{lemma}\label{lem:Feasibility}
$\{\tilde{r}_i\}$ is feasible with respect to LP~\eqref{eq:ReducedExpNewLP}.
\end{lemma}

\begin{proof} 
To prove the lemma it suffices (and we verify the sufficiency shortly) to show that $\{\tilde{r}_i\}$ meets the inequalities: 
\begin{equation}\label{eq:FirstIneq}
\sum_{i=1}^{Q} r_i-r_j+(P-Q+1)r_{Q+1}\geq k_j, \hspace{5pt} \forall 1\leq j\leq Q,\end{equation} and 
\begin{equation}\label{eq:LastIneq}
\sum_{i=1}^{Q} r_i+(P-Q)r_{Q+1}\geq k_j, \hspace{5pt} \forall Q< j\leq N-M.	
\end{equation} 	For every $1\leq j\leq Q$, the left-hand side (LHS) of the corresponding inequality in~\eqref{eq:FirstIneq} is the smallest in comparison with that of the rest of the corresponding inequalities in~\eqref{eq:ReducedExpNewLP}. This comes from the fact that in comparison with the LHS of the inequalities in~\eqref{eq:ReducedExpNewLP}, the LHS of the inequalities in \eqref{eq:FirstIneq} has the minimum number ($P-Q+1$) of the (larger) terms $\tilde{r}_i$, $i>Q$, or equivalently, the maximum number ($Q-1$) of the (smaller) terms $\tilde{r}_i$, $i\leq Q$ (by Lemma~\ref{lem:Increasingr}, $\tilde{r}_1\leq \tilde{r}_2\leq \dots\leq \tilde{r}_Q\leq \tilde{r}_{Q+1}=\dots=\tilde{r}_N$). Thus if the inequalities in~\eqref{eq:FirstIneq} hold true for $\{\tilde{r}_i\}$, then the rest of the inequalities in~\eqref{eq:ReducedExpNewLP} obviously hold true. The LHS of the inequalities in~\eqref{eq:LastIneq} are identical for every $Q< j\leq N-M$, and every such inequality holds true so long as the inequality with the largest right-hand side (RHS) ($k_{Q+1}$) holds true. (By~\eqref{eq:kis}, $k_{Q+1}\geq \dots\geq k_{N-M}$.) Thus, we can replace all the inequalities in~\eqref{eq:LastIneq} with one inequality: 
\begin{equation}\label{eq:NewLastIneq}
\sum_{i=1}^{Q} r_i+(P-Q)r_{Q+1}\geq k_{Q+1}.	
\end{equation} The rest of the proof is straightforward (and hence not included due to the lack of space) by showing that $\{\tilde{r}_i\}$ satisfies all the inequalities in~\eqref{eq:FirstIneq} and~\eqref{eq:NewLastIneq} with equality.
\end{proof}

\begin{lemma}\label{lem:Optimality}
$\{\tilde{r}_i\}$ is optimal with respect to LP~\eqref{eq:ReducedExpNewLP}.
\end{lemma}
 
\begin{proof} Consider the dual of LP~\eqref{eq:ReducedExpNewLP}: 
\begin{eqnarray}\label{eq:DualReducedExpNewLP}
\mathrm{maximize} && \sum_{j}\sum_{\mathcal{N}} k_j s_{\mathcal{N}}, \\ 
\nonumber \mbox{s.t.} && \sum_{j}\sum_{\cal{N}} s_{\mathcal{N}} \mathds{1}_{\{i\in\mathcal{N}\}} \leq 1, \hspace{5pt} \forall 1\leq i\leq N\\ 
\nonumber && \forall \{\mathcal{N}\subset [N]\setminus\{j\}: |\mathcal{N}|= P,\{i\}_{1\leq i<j}\in\mathcal{N}\}, \\ 
\nonumber && \forall 1\leq j\leq N-M,\\ 
\nonumber && (s_{\mathcal{N}}\geq 0).
\end{eqnarray} (We notice that $\mathcal{N}$ depends on $j$, yet we use the same notation $\mathcal{N}$, instead of $\mathcal{N}_j$, for the ease of exposition.)

We show that the duality gap with regards to LP~\eqref{eq:ReducedExpNewLP} and LP~\eqref{eq:DualReducedExpNewLP} is zero. To be more specific, we prove, by construction, there always exists a set ${\mathbb{S}}$ of $N$ subsets $\mathcal{N}$ such that $\{s_{\mathcal{N}}\}$ is feasible to LP~\eqref{eq:DualReducedExpNewLP} so long as $s_{{\mathcal{N}}}=\frac{1}{P}$, for every ${\mathcal{N}}\in{\mathbb{S}}$, and $s_{\mathcal{N}}=0$, for every $\mathcal{N}\notin\mathbb{S}$. By the structure of our construction process, $\{i\}$ ($1\leq i\leq N$) belongs to $P$ subsets ${\mathcal{N}}\in{\mathbb{S}}$. Thus, \[\sum_{j}\sum_{\mathcal{N}} s_{\mathcal{N}}\mathds{1}_{\{i\in\mathcal{N}\}}=1,\hspace{5pt} \forall 1\leq i\leq N.\] (Every inequality in~\eqref{eq:DualReducedExpNewLP} holds with equality.) Moreover, our choice of ${\mathbb{S}}$ has a partition $\{{\mathbb{S}}^{(1)},\dots,{\mathbb{S}}^{(Q+1)}\}$ such that (i) $|{\mathbb{S}}^{(j)}|=M+1$, $1\leq j\leq Q$, and $|{\mathbb{S}}^{(Q+1)}|=M-R+1$, and (ii) $\{j\}\notin {\mathcal{N}}$ and $\{i\}_{1\leq i<j}\in {\mathcal{N}}$, for every $\mathcal{N}\in {\mathbb{S}}^{(j)}$, $1\leq j\leq Q+1$. By (i) and (ii), it is obvious that
\begin{dmath}\label{eq:OptimalValueDual}
    \sum_{j}\sum_{\mathcal{N}} k_j s_{\mathcal{N}}=\sum_{j=1}^{Q}\frac{M+1}{P}k_j+\frac{M-R+1}{P}k_{Q+1}.
\end{dmath} By comparing~\eqref{eq:OptimalValuePrimal} and~\eqref{eq:OptimalValueDual}, it follows that the optimal values of the (primal) LP~\eqref{eq:ReducedExpNewLP} and (its dual) LP~\eqref{eq:DualReducedExpNewLP} are equal: 
\begin{equation}\label{eq:StrongDuality}
\sum_{i=1}^{N} \tilde{r}_i=\sum_{j}\sum_{\mathcal{N}} k_j s_{\mathcal{N}},
\end{equation} since $\frac{M+1}{P}=\frac{N-P}{P}$ and $\frac{M-R+1}{P}=\frac{Q(N-P)-N}{P}$. Thus, by the duality principle,~\eqref{eq:StrongDuality} proves the optimality of $\{\tilde{r}_{i}\}$. 

The rest of the proof proceeds by the construction of set $\mathbb{S}$ with properties (i) and (ii), defined earlier. Let $\mathbb{\tilde{S}}^{(j)}$, $1\leq j\leq Q$, be the set of all subsets $\{\mathcal{N}:\{i\}_{1\leq i<j}\in\mathcal{N},\{j\}\notin\mathcal{N},\{i\}_{j<i\leq Q}\in\mathcal{N}\}$, and $\mathbb{\tilde{S}}^{(Q+1)}$ be the set of all subsets $\{\mathcal{N}:\{i\}_{1\leq i\leq Q}\in\mathcal{N}\}$. For every $1\leq j\leq Q$ and every $0\leq m \leq M$, or $j=Q+1$ and every $0\leq m\leq M-R$, construct the (auxiliary) set $\mathbb{S}^{(j)}_{m+1}$ (in a recursive manner): \[\mathbb{S}^{(j)}_{m+1}=\mathbb{S}^{(j)}_{m}\cup \mathcal{N},\] for arbitrary $\mathcal{N}\in\mathbb{\tilde{S}}^{(j)}$, $\{Q< i_k\leq N\}_{1\leq k\leq P-Q+1}\in \mathcal{N}$, such that \[n_{i_1}(\mathbb{S}^{(j)}_{m})\leq \dots\leq n_{i_{P-Q+1}}(\mathbb{S}^{(j)}_{m})\leq\dots\leq n_{i_{N-Q}}(\mathbb{S}^{(j)}_{m}),\] where $n_i(\mathbb{S}^{(j)}_{m})$ is the number of subsets $\{\mathcal{N}:\mathcal{N}\in\mathbb{S}^{(j)}_{m}$, $\{i\}\in\mathcal{N}\}$; $\mathbb{S}^{(j)}_{0}=\mathbb{S}^{(j-1)}_{M+1}$, $1<j\leq Q+1$, and $\mathbb{S}^{(1)}_{0}=\emptyset$. (Such a subset $\mathcal{N}$ always exists since there exists a unique $\mathcal{N}\in\mathbb{\tilde{S}}^{(j)}$ such that $\{i\}_{i\in\mathcal{I}}\in\mathcal{N}$, for every $\mathcal{I}\subset [N]\setminus[Q]$ and $|\mathcal{I}|=P-Q+1$.) Now, we can construct partitions $\{\mathbb{S}^{(j)}\}$:
\begin{equation}\mathbb{S}^{(j)}=\left\{
\begin{array}{ll}
\mathbb{S}^{(j)}_{M+1}-\mathbb{S}^{(j)}_0, & 1\leq j\leq Q,\\ 
\mathbb{S}^{(j)}_{M-R+1}-\mathbb{S}^{(j)}_0, & j=Q+1.
\end{array} 
\right.\end{equation} By construction, both properties (i) and (ii) hold so long as $P$ and only $P$ subsets $\mathcal{N}\in\mathbb{S}$ exist such that $\{i\}\in\mathcal{N}$ (for every $i$). By the definition of $\mathbb{S}$, obviously $n_i(\mathbb{S})=P$, for every $1\leq i\leq Q$. Thus, it suffices to show $n_i(\mathbb{S})=P$, for every $Q<i\leq N$. Let $n^{(m)}_i$ be $n_i$ when $m$ subsets are chosen. By the structure of the construction, it is not hard to see at each step $m$ of the selection of one new subset for every $Q<i\leq N$, either $n_i$ increases by one or it does not change. (Thus, $0\leq n_i^{(m)}-n_i^{(m-1)}\leq 1$.) Since at the beginning of the process (when no subset is chosen) $n^{(0)}_i=0$, for every $i$, $n^{(m)}_i$, for every $1\leq m\leq N$, is either $n$ or $n-1$, for some $1\leq n\leq m$ (depending on $N$ and $M$). It suffices to show that $n_i^{(N)}=P$, for every $Q<i\leq N$. Let \[I(m,n)=|\{Q< i\leq N: n^{(m)}_i=n\}|.\] (Thus, $|\{Q< i\leq N: n^{(m)}_i=n-1\}\}|=N-Q-I(m,n)$.) It is easy to see that $n_i^{(N)}=P$ so long as $I(N,P)=N-Q$ (i.e., $N-Q-I(N,P)=0$). By analyzing the construction process step by step, it can be shown that \[I(m,n)=m(P-Q+1)-(n-1)(N-Q)+\xi(m),\] where 
\begin{dmath*}
	\xi(m)=\left\{
	\begin{array}{ll}
		0, & 1\leq m\leq (M+1)Q,\\ 
		(M+1)Q-m, & (M+1)Q<m\leq N.
	\end{array} 
	\right.
\end{dmath*} It is easy to verify $I(N,P)=N-Q$, and thus, for every $Q<i\leq N$, $n_i^{(N)}=P$. This completes the proof.
\end{proof}

The following result follows from the random packet distribution assumption (by the application of the law of large numbers), and is useful in the proof of Lemma~\ref{lem:NALP3and5}. 

\begin{lemma}\label{lem:SLLNResults}
For every $\mathcal{N}\subset[N]$, $|\mathcal{N}|=V$ ($0<V<N-M$), and $\mathcal{I}\subset \overline{\mathcal{N}}$, $|\mathcal{I}|=M$ ($0\leq M<N$), for any $\epsilon>0$, w.p.~approaching $1$ as $K\rightarrow\infty$, we have
\begin{dmath*}
	\Biggl|\frac{1}{K}\Biggl|\bigcap_{i\in\overline{\mathcal{N}}\setminus \mathcal{I}}\overline{X}_{i,\mathcal{I}}\Biggr|-Z_{M,V}\Biggr|<\epsilon,
\end{dmath*} where \[Z_{M,V}=\frac{(1-\alpha)^{-M-V}-(1-\alpha)^{-M}}{(1-\alpha)^{-N}-1}.\] 
\end{lemma}

\vspace{0.5cm}
\begin{proof} By~\eqref{eq:XiI}, it can be easily shown that
\[
\overline{X}_{i,\mathcal{I}}= \overline{X}_i\setminus(\cup_{j\in\mathcal{I}} X_j \setminus\cup_{j\in \overline{\mathcal{I}}} X_j),
\] and further, 
\begin{dmath}\label{eq:SetEquality}
	\Biggl|\bigcap_{i\in\overline{\mathcal{N}}\setminus\mathcal{I}} \overline{X}_{i,\mathcal{I}}\Biggr|=\Biggl|\bigcap_{i\in\overline{\mathcal{N}}\setminus\mathcal{I}} \overline{X}_i\Biggr|- \Biggl|\bigcup_{i\in\mathcal{I}} {X}_i\setminus\bigcup_{i\in\overline{\mathcal{I}}} X_i\Biggr|,
\end{dmath} where $\overline{\mathcal{I}} = [N]\setminus \mathcal{I}$. For every $x_n\in X$, 
\begin{dmath*} 
\nonumber \Pr\left\{x_n \in \bigcap_{i\in\overline{\mathcal{N}}\setminus\mathcal{I}} \overline{X}_i \right\}= \frac{(1-\alpha)^{-M-V}-1}{(1-\alpha)^{-N}-1},
\end{dmath*} and 
\begin{dmath*} 
\nonumber \Pr\left\{x_n \in \left\{\bigcup_{i\in\mathcal{I}} {X}_i\setminus\bigcup_{i\in\overline{\mathcal{I}}} X_i \right\}\right\}= \frac{(1-\alpha)^{-M}-1}{(1-\alpha)^{-N}-1}.
\end{dmath*} The followings hold true for any $\epsilon>0$, w.p.~approaching $1$ as $K\rightarrow\infty$. By the law of large numbers, 
\begin{dmath}\label{eq:Ineq1}
\Biggl|\frac{1}{K}\Biggl|\bigcap_{i\in\overline{\mathcal{N}}\setminus\mathcal{I}} \overline{X}_i\Biggr|	- \frac{(1-\alpha)^{-M-V}-1}{(1-\alpha)^{-N}-1}\Biggr|<\frac{\epsilon}{2}
\end{dmath} and 
\begin{dmath}\label{eq:Ineq2}
\Biggl|\frac{1}{K}\Biggl|\bigcup_{i\in\mathcal{I}} {X}_i\setminus\bigcup_{i\in\overline{\mathcal{I}}} X_i\Biggr|	- \frac{(1-\alpha)^{-M}-1}{(1-\alpha)^{-N}-1}\Biggr|<\frac{\epsilon}{2}.
\end{dmath} Thus, combining~\eqref{eq:Ineq1} and~\eqref{eq:Ineq2} together with~\eqref{eq:SetEquality}, \begin{dmath*}
	\left|\frac{1}{K}\left|\bigcap_{i\in\overline{\mathcal{N}}\setminus\mathcal{I}} \overline{X}_{i,\mathcal{I}}\right|-\frac{(1-\alpha)^{-M-V}-(1-\alpha)^{-M}}{(1-\alpha)^{-N}-1}\right|<\epsilon,
\end{dmath*} w.p.~approaching $1$ as $K\rightarrow\infty$.
\end{proof}
 
\begin{lemma}\label{lem:NALP3and5}
$\{\tilde{r}_i\}$ is an approximate solution to LP~\eqref{eq:FinalExpNewLP} and LP~\eqref{eq:ExpNewLP}.
\end{lemma}

\begin{proof} Since the objective functions in LP~\eqref{eq:FinalExpNewLP} and LP~\eqref{eq:ExpNewLP} are identical, and the constraints in LP~\eqref{eq:FinalExpNewLP} are a subset of the constraints in LP~\eqref{eq:ExpNewLP}, the following observations are straightforward:
(i) if $\{r_i\}$ is feasible to LP~\eqref{eq:ExpNewLP}, then it is feasible to LP~\eqref{eq:FinalExpNewLP}, and (ii) if $\{r_i\}$ is near-optimal to LP~\eqref{eq:FinalExpNewLP}, then it is near-optimal to LP~\eqref{eq:ExpNewLP}. Thus, it suffices to show $\{\tilde{r}_i\}$ is feasible to LP~\eqref{eq:ExpNewLP} (Lemma~\ref{lem:FeasibilityFinal}) and near-optimal to LP~\eqref{eq:FinalExpNewLP} (Lemma~\ref{lem:OptimalityFinal}). 
\end{proof}

\begin{lemma}\label{lem:FeasibilityFinal}
$\{\tilde{r}_i\}$ is feasible with respect to LP~\eqref{eq:ExpNewLP}.
\end{lemma}

\begin{proof} The feasibility follows immediately so long as $\{\tilde{r}_i\}$ meets the inequalities:
\begin{equation*}
\frac{1}{K}\sum_{i\in\cal{N}} r_i\geq \hspace{-5pt}\max_{{\scriptsize \begin{array}{c} \mathcal{I}\subset \mathcal{\overline{N}},\\ |\mathcal{I}|=M\end{array}}}\hspace{-2.5pt}\frac{1}{K}\Biggl|\bigcap_{i\in {\cal \overline{N}}\setminus\mathcal{I}}{\overline{X}}_{i,\mathcal{I}}\Biggr|, \hspace{5pt}\forall \{\mathcal{N}:1\leq |\mathcal{N}|\leq P\}.
\end{equation*} By Lemma~\ref{lem:Increasingr},
\[
\sum_{i\in \mathcal{N}} \tilde{r}_i\geq \sum_{i=1}^{V} \tilde{r}_i, \hspace{5pt} \forall\{\mathcal{N}:|\mathcal{N}|=V\}.
\] Thus, it suffices to show 
\begin{equation*}
\sum_{i=1}^{V} \tilde{r}_i\geq \max_{{\scriptsize \begin{array}{c} \mathcal{I}\subset [N]\setminus [V],\\ |\mathcal{I}|=M\end{array}}}\Biggl|\bigcap_{i\in [N]\setminus\mathcal{I}\cup [V]}{\overline{X}}_{i,\mathcal{I}}\Biggr|, \hspace{5pt} \forall 1\leq V\leq P.
\end{equation*}

The following is true (by the result of Lemma~\ref{lem:SLLNResults}) w.p.~approaching $1$ as $K\rightarrow\infty$. For any $\epsilon>0$, 
\begin{equation*}
\Biggl|\max_{{\scriptsize \begin{array}{c} \mathcal{I}\subset \mathcal{\overline{N}},\\ |\mathcal{I}|=M\end{array}}}\frac{1}{K}\Biggl|\bigcap_{i\in {\cal \overline{N}}\setminus\mathcal{I}}{\overline{X}}_{i,\mathcal{I}}\Biggr| -Z_{M,V}\Biggr|<\epsilon,\hspace{5pt}\forall \{\mathcal{N}: |\mathcal{N}|=V\}.
\end{equation*} Thus, we need to show
\begin{equation}\label{eq:FundamentalBound}
\frac{1}{K}\sum_{i=1}^{V} \tilde{r}_i> Z_{M,V}+\epsilon, \hspace{5pt} \forall 1\leq V\leq P
\end{equation} w.p.~approaching $1$ as $K\rightarrow\infty$. We consider two cases: (i) $V>Q$ and (ii) $V\leq Q$. In Case (i), 
\begin{dmath}\label{eq:SumTilderCase1}
\sum_{i=1}^{V} \tilde{r}_i=\sum_{j=1}^{Q}k_j \left(\frac{V-P}{P}\right) +k_{Q+1}\left(\frac{V+Q(P-V)}{P}\right),
\end{dmath} and in Case (ii),
\begin{dmath}\label{eq:SumTilderCase2}
\sum_{i=1}^{V} \tilde{r}_i=\sum_{j=1}^{V}k_j\left(\frac{V-P}{P}\right)+\sum_{j=V+1}^{Q}k_j\left(\frac{V}{P}\right)+k_{Q+1}\left(\frac{V+V(P-Q)}{P}\right).
\end{dmath}

Let $\epsilon_j$, $1\leq j\leq Q$, be equal to $\epsilon$ in Case (i) and Case (ii), and $\epsilon_{Q+1}$ be equal to $\frac{Q(P-V)}{Q(P-V)+V}\epsilon$ or $\frac{P-Q}{P-Q+1}\epsilon$ in Case (i) or Case (ii), respectively. By the result of Lemma~\ref{lem:SLLNResults}, 
\begin{dmath}\label{eq:kappaSLLN}
\frac{k_j}{K}> Z_{M,P} -\epsilon_j,
\end{dmath} w.p.~approaching $1$, when $K\rightarrow\infty$. By combining~\eqref{eq:SumTilderCase1} or~\eqref{eq:SumTilderCase2} together with~\eqref{eq:kappaSLLN}, we get 
\begin{dmath}\label{eq:CombiningSLLNandSumTilde}
\frac{1}{K}\sum_{i=1}^{V}	\tilde{r}_i> \frac{V}{P}Z_{M,P}+\epsilon,
\end{dmath} in Case (i) or (ii), respectively. By comparing~\eqref{eq:FundamentalBound} and~\eqref{eq:CombiningSLLNandSumTilde}, one can see we need to show 
\begin{equation}\label{eq:VPIneq}
\frac{V}{P}>\frac{Z_{M,V}}{Z_{M,P}}, \hspace{5pt} \forall 1\leq V\leq P
\end{equation} w.p.~approaching $1$ as $K\rightarrow\infty$ (for every $0<\alpha<1$). By substituting $Z_{M,V}$ and $Z_{M,P}$ into~\eqref{eq:VPIneq}, we get
\begin{equation}\label{eq:VPfInequality}
\frac{V}{P}>\varphi(\alpha),
\end{equation} where 
\[
\varphi(\alpha)=\frac{(1-\alpha)^{P-V}-(1-\alpha)^{P}}{1-(1-\alpha)^{P}}.
\] It is easy to see $\varphi(1)=0$ and $\varphi(\alpha)\rightarrow\frac{V}{P}$ as $\alpha\rightarrow 0$. By definition, $\frac{V}{P}\leq 1$. Thus,~\eqref{eq:VPfInequality} holds so long as $\varphi'(\alpha)< 0$, for every $0<\alpha<1$, where $\varphi'(\alpha)$ is the derivative of the function $\varphi(\alpha)$ with respect to $\alpha$ (i.e., $\varphi(\alpha)$ is decreasing, bounded from above by $\frac{V}{P}$). It is easy to see $\varphi'(\alpha)< 0$ so long as 
\begin{dmath}\label{eq:VPInverse}
\frac{1}{V}-\frac{1}{P}> \frac{(1-\alpha)^{V}}{V}-\frac{(1-\alpha)^{P}}{P}.
\end{dmath} Since 
\[
\frac{1}{n}>\frac{(1-\alpha)^n}{n},
\] for every $n>0$ and every $0<\alpha<1$,~\eqref{eq:VPInverse} holds so long as
\[
\frac{1}{n}-\frac{1}{n+1}>\frac{(1-\alpha)^{n}}{n}-\frac{(1-\alpha)^{n+1}}{n+1},
\] or equivalently, 
\begin{dmath}\label{eq:2ndIneq}
\frac{1-(1-\alpha)^n}{1-(1-\alpha)^{n+1}}> \frac{n}{n+1},
\end{dmath} for every $n>0$ and every $0<\alpha<1$. Let 
\[
\gamma(\alpha)=\frac{1-(1-\alpha)^n}{1-(1-\alpha)^{n+1}}.
\] Since $\gamma(1)=1$ and $\gamma(\alpha)\rightarrow\frac{n}{n+1}$, as $\alpha\rightarrow 0$,~\eqref{eq:2ndIneq} holds so long as $\gamma'(\alpha)>0$, for every $\alpha$ (i.e., $\gamma(\alpha)$ is increasing, bounded from below by $\frac{n}{n+1}$). It is easy to see $\gamma'(\alpha)>0$ so long as $\sum_{m=1}^{n}(1-\alpha)^{m}<n$, which obviously holds true for every $0<\alpha<1$. 
\end{proof}

\begin{lemma}\label{lem:OptimalityFinal}
$\{\tilde{r}_i\}$ is near-optimal with respect to LP~\eqref{eq:FinalExpNewLP}.	
\end{lemma}

\begin{proof} The dual of LP~\eqref{eq:FinalExpNewLP} can be written as: 
\begin{eqnarray}\label{eq:DualFinalExpNewLP}
\mathrm{maximize} && \sum_{\mathcal{N}} \max_{i\in\mathcal{\overline{N}}}\left|\overline{X}_{i,\mathcal{\overline{N}}\setminus \{i\}}\right| s_{\mathcal{N}}, \\ 
\nonumber \mbox{s.t.} && \sum_{\cal{N}} s_{\mathcal{N}} \mathds{1}_{\{i\in\mathcal{N}\}} \leq 1, \hspace{5pt} \forall 1\leq i\leq N\\ 
\nonumber && \forall \{\mathcal{N}\subset [N]:|\mathcal{N}|= P\}, \\ 
\nonumber && (s_{\mathcal{N}}\geq 0).
\end{eqnarray} Let $r^{*}_{K}$ be the optimal value of LP~\eqref{eq:FinalExpNewLP}. By the definition of the near-optimality, we require to show
\begin{equation}\label{eq:LastIneq00}
\frac{1}{K}\sum_{i=1}^{N}\tilde{r}_i< \frac{1}{K} r^{*}_{K}+\epsilon,
\end{equation} for any $\epsilon>0$. To do so, we use the set $\mathbb{S}=\{\mathbb{S}^{(1)},\dots,\mathbb{S}^{(Q+1)}\}$ which we previously constructed in the proof of Lemma~\ref{lem:Optimality}, and set $s_{\mathcal{N}}=\frac{1}{P}$, for every $\mathcal{N}\in\mathbb{S}$, and $s_{\mathcal{N}}=0$, for every $\mathcal{N}\notin\mathbb{S}$. Since LP~\eqref{eq:DualFinalExpNewLP} and LP~\eqref{eq:DualReducedExpNewLP} have identical constraints, $\{s_{\mathcal{N}}\}$, which was shown to be feasible with respect to LP~\eqref{eq:DualReducedExpNewLP}, is feasible with respect to LP~\eqref{eq:DualFinalExpNewLP}. Thus, by the duality principle, 
\begin{dmath}\label{eq:LastIneq0}
r^{*}_{K} \geq \sum_{\mathcal{N}} \max_{i\in\mathcal{\overline{N}}}\left|\overline{X}_{i,\overline{\mathcal{N}}\setminus\{i\}}\right| s_{\mathcal{N}} = \frac{1}{P}\sum_{j=1}^{Q+1}\sum_{\mathcal{N}\in\mathbb{S}^{(j)}} \max_{i\in \mathcal{\overline{N}}}\left|\overline{X}_{i,\overline{\mathcal{N}}\setminus\{i\}}\right|.
\end{dmath} The following results hold w.p.~approaching $1$ as $K\rightarrow\infty$ (by the results of Lemma~\ref{lem:SLLNResults}). For any $\{\epsilon_j>0\}$ and every $\mathcal{N}\in\mathbb{S}^{(j)}$,  
\begin{equation}\label{eq:LastIneq2}
\frac{1}{K} \max_{i\in\mathcal{\overline{N}}}\left|\overline{X}_{i,\overline{\mathcal{N}}\setminus\{i\}}\right|>Z_{M,P}-\frac{\epsilon_j}{2}, \hspace{5pt}  \forall 1\leq j\leq Q+1.
\end{equation} Since $|\mathbb{S}^{(j)}|=M+1$, $1\leq j\leq Q$, and $|\mathbb{S}^{(Q+1)}|=M-R+1$, combining~\eqref{eq:LastIneq0} and~\eqref{eq:LastIneq2} we can write
\begin{dmath}\label{eq:LastIneq3}
\frac{1}{K} \sum_{\mathcal{N}} \max_{i\in\mathcal{\overline{N}}}\left|\overline{X}_{i,\overline{\mathcal{N}}\setminus\{i\}}\right| s_{\mathcal{N}}> \frac{N}{P}Z_{M,P}-\sum_{j=1}^{Q} \left(\frac{M+1}{P}\right)\frac{\epsilon_j}{2} - \left(\frac{M-R+1}{P}\right)\frac{\epsilon_{Q+1}}{2}.
\end{dmath} Similarly, 
\begin{equation}\label{eq:LastIneq4}
\frac{k_j}{K}<Z_{M,P}+\frac{\epsilon_j}{2},\hspace{5pt} \forall 1\leq j\leq Q+1.
\end{equation} By combining~\eqref{eq:LastIneq4} together with~\eqref{eq:OptimalValuePrimal}, we get
\begin{dmath}\label{eq:LastIneq5}
\frac{1}{K}\sum_{i=1}^{N}\tilde{r}_i<\frac{N}{P}Z_{M,P}+\sum_{j=1}^{Q}\left(\frac{N-P}{P}\right)\frac{\epsilon_j}{2}+\left(\frac{N+Q(P-N)}{P}\right)\frac{\epsilon_{Q+1}}{2}
\end{dmath} By combining~\eqref{eq:LastIneq5} and~\eqref{eq:LastIneq3} together with~\eqref{eq:LastIneq0}, one can see~\eqref{eq:LastIneq00} holds so long as 
\begin{dmath}\label{eq:LastIneq6}
\epsilon> \sum_{j=1}^{Q} \left(\frac{M+1}{P}\right)\epsilon_j + \left(\frac{M-R+1}{P}\right)\epsilon_{Q+1}.
\end{dmath} The RHS of~\eqref{eq:LastIneq6} can be made arbitrarily close to $0$, and this completes the proof.
\end{proof}

\subsection{Proof of Theorem~\ref{thm:SCOverConstrainedReducedLPSolution}}\label{subsec:ExactThm}
LP~\eqref{eq:FinalExpNewLP} can be rewritten as (when $M=1$):
\begin{eqnarray}\label{eq:SCFinalExpNewLP}
\mathrm{minimize} && \sum_{i=1}^{N} r_i, \\ \nonumber
\mbox{s.t.} && \hspace{-0.25cm}\sum_{i\in\mathcal{N}_{m,n}} r_i\geq k_{m,n}, \hspace{5pt} \forall 1\leq m<n\leq N,
\end{eqnarray} where $k_{m,n}$ is given by~\eqref{eq:kijdef}, and $\mathcal{N}_{m,n}=[N]\setminus \{m,n\}$. 

The following two lemmas are useful in the proof of the theorem. (The proofs are straightforward and hence omitted). 

\begin{lemma}\label{lem:SCkPairs}
For every $1\leq m_1,m_2,n_1,n_2\leq N$, 
\begin{dmath}\label{eq:ki1i2j1j2} k_{m_1,n_1}-k_{m_1,n_2}=k_{m_2,n_1}-k_{m_2,n_2},
\end{dmath} so long as $\{m_1,m_2\}\cap \{n_1,n_2\}=\emptyset$.
\end{lemma}


\begin{lemma}\label{lem:SCkPairs2}
For every $1\leq n<m_1\leq m_2\leq N$, 
\begin{dmath}\label{eq:kij1j2}	
k_{m_1,n}-k_{m_2,n}\geq k_{n,m_1}-k_{n,m_2}.	
\end{dmath}	
\end{lemma}


For every $1\leq m<n\leq N$, let $\lambda_{m,n}$ be defined as
\begin{dmath}\label{eq:SClambdadef}\lambda_{m,n}=\left\{
\begin{array}{l}
k_{m,N}-k_{1,N}+k_{1,N-Q} \\ \hspace{0.5cm} -k_{N-Q,N}+k_{n,N}, \hspace{5pt}\hfill m< n\leq N-Q,\\ 
k_{m,N}-k_{1,N}+k_{1,n}, \hfill m\leq N-Q< n,\\
k_{1,m}-k_{1,N}+k_{N-Q,N}\\ \hspace{0.5cm} -k_{1,N-Q}+k_{1,n}, \hspace{5pt}
\hfill N-Q<m<n.
\end{array} 
\right.\end{dmath} By applying the results of lemmas~\ref{lem:SCkPairs} and~\ref{lem:SCkPairs2}, the following result can then be shown. (The proof is omitted due to the lack of space.)

\begin{lemma}\label{lem:SCLambdakComp}
For every $1\leq m<n\leq N$, \begin{equation*}\lambda_{m,n}\geq k_{m,n}.\end{equation*}
\end{lemma}

\vspace{0.25cm}
We now construct LP~\eqref{eq:SCReducedExpNewLP} (by over-constraining LP~\eqref{eq:SCFinalExpNewLP}): 
\begin{eqnarray}\label{eq:SCReducedExpNewLP}
\mathrm{minimize} && \sum_{i=1}^{N} r_i, \\ 
\nonumber \mbox{s.t.} && \hspace{-0.25cm}\sum_{i\in\mathcal{N}_{m,n}} r_i\geq \lambda_{m,n}, \hspace{5pt} \forall 1\leq m<n\leq N.
\end{eqnarray} (By the result of Lemma~\ref{lem:SCLambdakComp}, it is easy to see the constraints in~\eqref{eq:SCReducedExpNewLP} are stronger than those in~\eqref{eq:SCFinalExpNewLP}.) 

We now prove the theorem in two steps: (i) we show $\{\tilde{r}_i\}$ gives an exact solution to LP~\eqref{eq:SCReducedExpNewLP}, and (ii) we show $\{\tilde{r}_i\}$ is an exact solution to LP~\eqref{eq:FinalExpNewLP} and LP~\eqref{eq:ExpNewLP}.

\begin{lemma}\label{lem:SCReducedExpNewLPSolution}
$\{\tilde{r}_i\}$ is an exact solution to LP~\eqref{eq:SCReducedExpNewLP}.	
\end{lemma}

\begin{proof} We prove the feasibility and the optimality of $\{\tilde{r}_i\}$ to LP~\eqref{eq:SCReducedExpNewLP} in lemmas~\ref{lem:SCFeasibility} and~\ref{lem:SCOptimality}, respectively.
\end{proof}

\begin{lemma}\label{lem:SCFeasibility}
$\{\tilde{r}_i\}$ is feasible with respect to LP~\eqref{eq:SCReducedExpNewLP}.	
\end{lemma}

\begin{proof} Since $\lambda_{j,N}=k_{j,N}$ ($1\leq j\leq N-Q$) and $\lambda_{1,j}=k_{1,j}$ ($N-Q\leq j< N$), by~\eqref{eq:SClambdadef} it is not hard to see that $\{\tilde{r}_i\}$ meets (with equality) every inequality in LP~\eqref{eq:SCReducedExpNewLP} so long as $\{\tilde{r}_i\}$ meets (with equality) the $N$ inequalities: 
\begin{eqnarray*}
\sum_{i\in \mathcal{N}_{j,N}} r_i &\geq & \lambda_{j,N}, \hspace{5pt} \forall 1\leq j\leq N-Q,\\
\sum_{i\in \mathcal{N}_{1,j}} r_i &\geq & \lambda_{1,j}, \hspace{5pt} \forall N-Q\leq j<N,	
\end{eqnarray*} Furthermore, the proof of the latter is straightforward (and hence omitted). 
\end{proof}

\begin{lemma}\label{lem:SCOptimality}
$\{\tilde{r}_i\}$ is optimal with respect to LP~\eqref{eq:SCReducedExpNewLP}.	
\end{lemma}

\begin{proof} The dual of LP~\eqref{eq:SCReducedExpNewLP} is given by
\begin{eqnarray}\label{eq:SCDualReducedExpNewLP}
\mathrm{maximize} && \sum_{m,n} \lambda_{m,n} s_{\mathcal{N}_{m,n}}, \\ 
\nonumber \mbox{s.t.} && \sum_{m,n} s_{\mathcal{N}_{m,n}} \mathds{1}_{\{i\notin\{m,n\}\}} \leq 1, \hspace{5pt} \forall 1\leq i\leq N\\ 
\nonumber && \forall 1\leq m<n\leq N,\\ 
\nonumber && (s_{\mathcal{N}_{m,n}}\geq 0).
\end{eqnarray} Similar to the proof of Lemma~\ref{lem:Optimality}, we construct the set $\mathbb{S} = \{{\mathbb{S}}^{(1)},\dots,{\mathbb{S}}^{(Q+1)}\}$ of $N$ subsets $\mathcal{N}_{m,n}$ such that $\{s_{\mathcal{N}_{m,n}}\}$ is feasible to LP~\eqref{eq:SCDualReducedExpNewLP}, where $s_{\mathcal{N}_{m,n}}=\frac{1}{P}$, for every $\mathcal{N}_{m,n}\in\mathbb{S}$, and $s_{\mathcal{N}_{m,n}}=0$, for every $\mathcal{N}_{m,n}\notin \mathbb{S}$. Considering four cases (depending on $Q$ and $R$), we construct the partitions $\{\mathbb{S}^{(j)}\}$:  
\begin{itemize}
\item[(i)] $Q$ odd, $R=0$: $\mathbb{S}^{(j)}=\{\mathcal{N}_{2j-1,N-j+1},$ $\mathcal{N}_{2j,N-j+1}\}$, $1\leq j\leq \frac{Q+1}{2}$, and $\mathbb{S}^{(j)}=\{\mathcal{N}_{2j-Q-2,N-j+1},$ $\mathcal{N}_{2j-Q-1,N-j+1}\}$, $\frac{Q+1}{2}< j\leq Q+1$.
\item[(ii)] $Q$ odd, $R=1$: $\mathbb{S}^{(j)}=\{\mathcal{N}_{2j-1,N-j+1},$ $\mathcal{N}_{2j,N-j+1}\}$, $1\leq j\leq \frac{Q+1}{2}$, and $\mathbb{S}^{(j)}=\{\mathcal{N}_{2j-Q-2,N-j+1},$ $\mathcal{N}_{2j-Q-1,N-j+1}\}$, $\frac{Q+1}{2}< j\leq Q$, and $\mathbb{S}^{(Q+1)}=\{\mathcal{N}_{Q,N-Q}\}$.
\item[(iii)] $Q$ even, $R=0$:	$\mathbb{S}^{(j)}=\{\mathcal{N}_{2j-1,N-j+1},$ $\mathcal{N}_{2j,N-j+1}\}$, $1\leq j\leq \frac{Q}{2}$, and $\mathbb{S}^{(\frac{Q}{2}+1)}=\{\mathcal{N}_{Q+1,N-\frac{Q}{2}},$ $\mathcal{N}_{1,N-\frac{Q}{2}}\}$, and $\mathbb{S}^{(j)}=\{\mathcal{N}_{2j-Q-2,N-j},$ $\mathcal{N}_{2j-Q-1,N-j}\}$, $\frac{Q}{2}+1< j\leq Q+1$.
\item[(iv)] $Q$ even, $R=1$:	$\mathbb{S}^{(j)}=\{\mathcal{N}_{2j-1,N-j+1},$ $\mathcal{N}_{2j,N-j+1}\}$, $1\leq j\leq \frac{Q}{2}$, and $\mathbb{S}^{(\frac{Q}{2}+1)}=\{\mathcal{N}_{Q+1,N-\frac{Q}{2}},$ $\mathcal{N}_{1,N-\frac{Q}{2}}\}$, and $\mathbb{S}^{(j)}=\{\mathcal{N}_{2j-Q-2,N-j},$ $\mathcal{N}_{2j-Q-1,N-j}\}$, $\frac{Q}{2}+1< j\leq Q$, and $\mathbb{S}^{(Q+1)}=\{\mathcal{N}_{Q,N-Q}\}$.
\end{itemize} In each case (i)--(iv), it is easy to see $\{i\}$, $1\leq i\leq N$, belongs to $P$ subsets $\mathcal{N}_{m,n}\in\mathbb{S}$. Thus, 
\begin{equation*}\sum_{m,n} s_{\mathcal{N}_{m,n}} \mathds{1}_{\{i\notin\{m,n\}\}} = 1,\hspace{5pt} \forall 1\leq i\leq N.
\end{equation*} This confirms the feasibility of $\{s_{\mathcal{N}_{m,n}}\}$. Then, by the duality principle, it suffices to show that \[\sum_{i=1}^{N}\tilde{r}_i=\sum_{m,n}\lambda_{m,n}s_{\mathcal{N}_{m,n}}.\] 

We only give the proof for the case (i) here (and the proofs for the other cases are similar). In the case (i), by our choice of $\mathbb{S}$, $\sum_{m,n}\lambda_{m,n}s_{\mathcal{N}_{m,n}}$ equals to \vspace{-0.25cm}
\begin{dmath}\label{eq:SCCaseI1}
\sum_{j=1}^{\frac{Q+1}{2}}\frac{1}{P}\left(\lambda_{2j-1,N-j+1}+\lambda_{2j,N-j+1}\right) + \sum_{j=\frac{Q+3}{2}}^{Q+1}\frac{1}{P}\left(\lambda_{2j-Q-2,N-j+1}+\lambda_{2j-Q-1,N-j+1}\right)
\end{dmath} By using~\eqref{eq:SClambdadef},~\eqref{eq:SCCaseI1} can be written as
\begin{equation*}
\frac{2-2Q}{P}k_{1,N}+\sum_{i=2}^{Q+1}\frac{2}{P}k_{i,N} + \sum_{i=N-Q}^{N-1}\frac{2}{P}k_{1,i},
\end{equation*} which equals to $\sum_{i=1}^{N}\tilde{r}_i$, since in the case (i) $R=0$ (by assumption) and thus $N-Q-1=Q+1$ and $P=2Q$.
\end{proof}

 

\begin{lemma}\label{lem:SCNALP3and5}
$\{\tilde{r}_i\}$ is an exact solution to LP~\eqref{eq:SCFinalExpNewLP} and LP~\eqref{eq:ExpNewLP}.
\end{lemma}

\begin{proof} By a similar argument as in the proof of Lemma~\ref{lem:NALP3and5}, it suffices to show the feasibility and optimality of $\{\tilde{r}_i\}$ with respect to LP~\eqref{eq:ExpNewLP} (Lemma~\ref{lem:SCFeasibilityFinal}) and LP~\eqref{eq:SCFinalExpNewLP} (Lemma~\ref{lem:SCOptimalityFinal}), respectively. 
\end{proof}

\begin{lemma}\label{lem:SCFeasibilityFinal}
$\{\tilde{r}_i\}$ is feasible with respect to LP~\eqref{eq:ExpNewLP}.
\end{lemma}

\begin{proof} We assume $K\rightarrow\infty$, and the results of Lemma~\ref{lem:SLLNResults} hold w.p.~approaching $1$, for any $\epsilon>0$. We need to show that $\{\tilde{r}_i\}$ meets the inequalities:
\begin{equation*}
\frac{1}{K}\sum_{i\in\cal{N}} r_i\geq \hspace{-5pt}\max_{{\scriptsize \begin{array}{c} j\subset \mathcal{\overline{N}}\end{array}}}\hspace{-2.5pt}\frac{1}{K}\Biggl|\bigcap_{i\in {\cal \overline{N}}\setminus \{j\}}\hspace{-2.5pt}{\overline{X}}_{i,\{j\}}\Biggr|,\forall \{\mathcal{N}: 1\leq |\mathcal{N}|\leq P\}.
\end{equation*} From Lemma~\ref{lem:SLLNResults}, for every $\mathcal{N}\subset[N]$, $|\mathcal{N}|=V$, it follows 
\begin{equation}\label{eq:SCmaxZMV}
\left|\max_{{\scriptsize \begin{array}{c} j\subset \mathcal{\overline{N}}\end{array}}}\frac{1}{K}\Biggl|\bigcap_{i\in {\cal \overline{N}}\setminus\{j\}}{\overline{X}}_{i,\{j\}}\Biggr|-Z_{1,V}\right|<\epsilon, \hspace{5pt} \forall 1\leq V\leq P.
\end{equation} Thus, it suffices to show 
\begin{eqnarray}\label{eq:SCrZ1V}
&& \frac{1}{K}\sum_{i\in\cal{N}} \tilde{r}_i>Z_{1,V}+\epsilon,\hspace{5pt} \forall 1\leq V\leq P,
\end{eqnarray} for every $\mathcal{N}\subset[N]$, $|\mathcal{N}|=V$. Moreover, 
\begin{equation}\label{eq:SCkmndev}
\left|\frac{k_{m,n}}{K}-Z_{1,P}\right|<\frac{P}{V}\epsilon, \hspace{5pt} \forall 1\leq m<n\leq N
\end{equation} By combining~\eqref{eq:SCkmndev} with~\eqref{eq:SCtilder} and~\eqref{eq:SCOptimalValuePrimal}, we can write 
\begin{eqnarray}\label{eq:SCrZ1P}
&& \frac{1}{K}\sum_{i\in\cal{N}} \tilde{r}_i>\frac{V}{P}Z_{1,P}+\epsilon, \hspace{5pt} \forall 1\leq V\leq P,
\end{eqnarray} for every $\mathcal{N}\subset[N]$, $|\mathcal{N}|=V$. From~\eqref{eq:SCrZ1V} and~\eqref{eq:SCrZ1P}, one can see that the proof of the lemma is complete so long as
\begin{equation}\label{eq:SCVPZ1VZ1P}
\frac{V}{P}>\frac{Z_{1,V}}{Z_{1,P}}, \hspace{5pt} \forall 1\leq V\leq P.	
\end{equation} Furthermore,~\eqref{eq:SCVPZ1VZ1P} is a special case of~\eqref{eq:VPIneq}, which was previously shown in the proof of Lemma~\ref{lem:FeasibilityFinal}.
\end{proof}	

\begin{lemma}\label{lem:SCOptimalityFinal}
$\{\tilde{r}_i\}$ is optimal with respect to LP~\eqref{eq:SCFinalExpNewLP}.
\end{lemma}

\begin{proof} The proof follows the exact same line as in the proof of Lemma~\ref{lem:SCOptimality} (and hence omitted to avoid repetition) since: (i) $\lambda_{m,n}=k_{m,n}$, for every $m\leq N-Q< n$ (by \eqref{eq:SClambdadef}), and (ii) $m\leq N-Q<n$, for every $\mathcal{N}_{m,n}\in\mathbb{S}$.	
\end{proof}

\bibliographystyle{IEEEtran}
\bibliography{CDERefs}

\end{document}